\documentclass[conference]{IEEEtran}
\usepackage{verbatim}
\usepackage{epsfig,bbm}
\usepackage[colorlinks,bookmarksopen,bookmarksnumbered,citecolor=red,urlcolor=red,]{hyperref}
\usepackage{CJK}
\usepackage{indentfirst}
\usepackage{multirow}
\usepackage{epstopdf}
\usepackage{graphicx}
\usepackage{footmisc}
\graphicspath{{./figures/}}
\usepackage{amsfonts}
\usepackage{mathrsfs}
\usepackage{setspace}
\usepackage{amsmath}
\usepackage{algorithm,algorithmic,amsbsy,amsmath,amssymb,epsfig,bbm,mathrsfs, bbm} 
\usepackage{amsthm}
\usepackage{verbatim}
\usepackage[subfigure]{graphfig}

\newtheorem{theorem}{Theorem}
\newtheorem{corollary}{Corollary}

\begin{document}

\title{Performance Analysis of Ambient RF Energy Harvesting: A Stochastic Geometry Approach}
 \author{Ian Flint$^{\dagger}$, Xiao Lu$^{\ddagger}$, Nicolas Privault$^{\dagger}$, Dusit Niyato$^{\ddagger}$, and Ping Wang$^{\ddagger}$	\\
    ~$^{\dagger}$ School of Physical $\&$ Mathematical Sciences, Nanyang Technological University, Singapore  \\
		~$^{\ddagger}$ School of Computer Engineering, Nanyang Technological University, Singapore \\
   }

\markboth{}{Shell \MakeLowercase{\textit{et al.}}: Bare Demo of
IEEEtran.cls for Journals}

\maketitle

\begin{abstract}
Ambient RF (Radio Frequency) energy harvesting technique has recently been proposed as a potential solution to provide proactive energy replenishment for wireless devices. This paper aims to analyze the performance of a battery-free wireless sensor powered by ambient RF energy harvesting using a stochastic-geometry approach. Specifically, we consider a random network model in which ambient RF sources are distributed as a Ginibre $\alpha$-determinantal point process which recovers the Poisson point process when $\alpha$ approaches zero. We characterize the expected RF energy harvesting rate. We also perform a worst-case study which derives the upper bounds of both power outage and transmission outage probabilities.
Numerical results show that our upper bounds are accurate and that better performance is achieved when the distribution of ambient sources exhibits stronger repulsion.

\end{abstract}

\emph{Index terms- RF energy harvesting, sensor networks, determinantal point process, Poisson point process, Ginibre model}  .

\section{Introduction}
 
Ambient RF energy harvesting technique \cite{XLuSurvey,X.2014Lu,X2014Lu} offers the capability of converting the received RF signals into electricity. Therefore, it has recently emerged as an alternative method to operate low-power devices~\cite{H2010Nishimoto}, such as wireless sensors. The fact that ambient RF energy harvesting aims to capture and recycle the environmental energy such as broadcast TV, radio and cellular signals, which are essentially free and universally present, makes this technique even more appealing. In this context, wireless devices powered by ambient RF energy are enabled for battery-free implementation, and has a potential for perpetual operation. Experiment with ambient RF energy harvesting in~\cite{A2009Sample} shows that $60\mu W$ is harvested from a TV tower that are $4.1km$ away. 
\cite{D2010Bouchouicha} measures the ambient RF power density from $680MHz$ to $3.5GHz$ and shows that the average power density from $1GHz$ to $3.5GHz$ is in the order of $63\mu W/m^{2}$. Detected $6.3km$ away from Tokyo Tower, the RF-to-DC conversion efficiency is demonstrated to be about $16\%$, $30\%$ and $41\%$ when the input power is $-15dBm$, $-10dBm$ and $-5dBm$, respectively~\cite{R2003Shigeta}.  

Geometry approaches have been applied to analyze RF energy harvesting performance in cellular network~\cite{K2014Huang}, relay network~\cite{I2014Krikidis} and cognitive radio network~\cite{S2013Lee}.~\cite{K2014Huang} investigates tradeoffs among transmit power and density of mobiles and wireless charging stations which both locate as homogeneous Poisson Point Process (PPP).~\cite{I2014Krikidis} studies the impact of cooperative density and relay selection in a large-scale network with transmitter-receiver pairs distributed as PPP. The authors in~\cite{S2013Lee} study a cognitive radio network where the primary and secondary networks are distributed as independent homogeneous PPPs. The secondary network is powered by the energy opportunistically harvested from nearby transmitters in the primary network. Under the outage probability requirements for both coexisting networks, the maximum throughput of the secondary network is analyzed. Prior literature mainly focuses on performance analysis on RF-powered wireless devices using PPP. 
The study in \cite{A1405.2013H} investigates cognitive and energy harvesting-based device-to-device (D2D) communication underlying cellular networks. The authors adopt independent PPPs to model the locations of BSs, cellular mobiles, and D2D devices. Specifically, two spectrum access policies designed for cellular BSs, namely, random spectrum access and prioritized spectrum access, are studied.
In this paper, we generalize the conventional PPP analytical framework to a $\alpha$-determinantal point process (DPP), where PPP is a special case when $\alpha$ approaches zero. We focus more specifically on the so-called Ginibre DPP which offers many advantages in terms of modeling capability and ease of simulation~\cite{DecreusefondFlintVergne}.




To address this issue, this paper analyzes the point-to-point transmission between an RF-powered sensor node and a data sink. The sensor node needs to harvests RF energy from ambient RF sources (e.g., cellular mobiles). The ambient RF sources are modeled with a Ginibre $\alpha$-DPP. The sensor transmits to the data sink using instant harvested RF energy. A power outage happens if the instant harvested energy fails to meet the circuit power consumption of the sensor. Moreover, if the minimum transmission rate requirement cannot be fulfilled, a transmission outage occurs. Based on the above model, we first derive the expression of expected RF energy harvesting rate, then characterize the upper bounds of both power outage probability and transmission outage probability in a closed form. The performance analysis provides a useful insight into the tradeoff among various network parameters.

The remainder of this paper is organized as follows. Section~II introduces the network model, geometry model of ambient RF sources and performance metrics. Section~III estimates the performance metrics of the sensor for both Ginibre $\alpha$-DPP and PPP modeling of ambient RF sources. The numerical results are demonstrated in Section~IV, followed by the conclusion in Section~V.

\section{System Model}

\subsection{Network Model}
\begin{figure}
\centering
\includegraphics[width=0.45\textwidth]{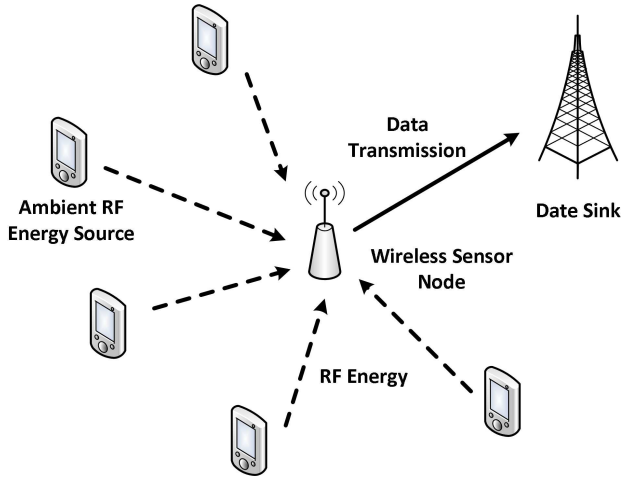}
\caption{A network model of ambient RF energy harvesting.} \label{fig:systemmodel}
\end{figure}

We consider a sensor node $i$ harvesting RF energy from ambient RF energy sources. We assume that the ambient RF energy sources are distributed as a Ginibre $\alpha$-DPP, which will be specified in detail in Section~III. We assume that the sensor node is solely powered by the harvested RF energy to supply its operations for data transmission. 
The sensor contains an energy harvester and an information transmitter, which are equipped with separated antennas, so that they can function independently and concurrently. In other words, the sensor is able to perform data transmission and RF energy harvesting simultaneously.  The instantaneous harvested energy is first used to operate the sensor circuit and then the surplus energy is provided for information transmission.

The RF energy harvesting rate by the sensor node $i$ from the RF energy source $k$ in a free-space channel $P^{i,k}_{\mathrm{H}}$ can be obtained based on the Friis equation~\cite{Visser2013} as follows:
\begin{equation}
	P^{i,k}_{\mathrm{H}}	=	\beta_{i} P^k_{\mathrm{S}} \frac{G^{k}_{\mathrm{S}} G^{i}_{\mathrm{H}} ( \lambda_{k} )^{2}}{(4\pi d_{i,k})^{2}}
\end{equation}
where $\beta_{i}$ is the RF-to-DC power conversion efficiency of the sensor node $i$. $P^k_{\mathrm{S}}$ is the transmit power of RF energy source $k$. $G^{k}_{\mathrm{S}}$ is the transmit antenna gain of RF energy source $k$. $G^{i}_{\mathrm{H}}$ is the receive antenna gain  of sensor node $i$. $\lambda_{k}$ is the wavelength emitted at RF energy source $k$. $d_{i,k}$ is the distance between the transmit antenna of RF energy source $k$ and the receiver antenna of sensor node $i$. Let ${\mathbf{x}}_i \in {\mathbb{R}}^2$ and ${\mathbf{x}}_k \in {\mathbb{R}}^2$ denote the coordinates of the sensor node $i$ and RF energy source $k$, respectively. The distance can be obtained from $d_{i,k} = \epsilon + \lVert	{\mathbf{x}}_i - {\mathbf{x}}_k	\rVert$, where $\epsilon$ is a fixed (small) parameter which ensures that the associated harvested RF power is finite in expectation. Physically, $\epsilon$ is the closest distance that the RF energy sources can be to the sensor node.

Then, the aggregated RF energy harvesting rate by sensor node $i$ can be obtained from
\begin{equation}
	P^i_{\mathrm{H}} = \sum_{ k \in {\mathcal{K}} } P^{i,k}_{\mathrm{H}}	=	\sum_{ k \in {\mathcal{K}} }	\beta_{i} P^k_{\mathrm{S}} \frac{G^{k}_{\mathrm{S}} G^{i}_{\mathrm{H}} ( \lambda_{k} )^{2}}{(4\pi 	(\epsilon + \lVert {\mathbf{x}}_i - {\mathbf{x}}_k \rVert)	)^{2}}
\end{equation}
where ${\mathcal{K}}$ is a random set consisting of all RF energy sources. We assume that $\mathcal{K}$ is a point process~\cite{Kallenberg}.


The sensor consumes a base circuit power, denoted as $P^i_{\mathrm{C}}$. Following practical models~\cite{G2009Miao}, the circuit power consumption of the sensor is assumed to be fixed.
The transmit power of sensor node $i$ is then given by $P^i_{\mathrm{T}} = \left[	P^i_{\mathrm{H}} - P^i_{\mathrm{C}} \right]^+$ where $[x]^+ = \max(0, x)$. The maximum transmission rate of sensor node $i$ is obtained as follows:
\begin{equation}
	C_i	=	W \cdot \log_2	\Big(	1 + h_i	\frac {	P^i_{\mathrm{T}}	}	{\sigma^2} 	\Big)
\end{equation}
where W is the transmission bandwidth. $h_i$ denotes the channel gain between the transmit antenna of sensor node $i$ and the receive antenna of data sink. $\sigma^2$ is the noise power.

\subsection{Geometric Modeling of Ambient RF Energy Sources}

As an extension of the Poisson setting, we model the locations of RF energy sources using a point process $\mathcal{K}$ on an observation window $O:=\mathcal{B}(0,R)$ which is the closed ball centered at the origin and of radius $R>0$. In other terms, $\mathcal{K}$ is an almost surely finite random collection of points inside $\mathcal{B}(0,R)$. We refer to~\cite{Kallenberg} and~\cite{DaleyVereJones} for the general theory of point processes. 

We let $\alpha=-\frac{1}{m}$ for $m\in\mathbb{N}^*$ and we focus on the Ginibre $\alpha$-DPP, which is a class of point processes well-suited for applications. The Ginibre process is a type of $\alpha$-DPP (see~\cite{ShiraiTakahashi} for definitions and technical results). 
The Ginibre process is defined by the so-called Ginibre kernel given by
\begin{equation}
\label{eq:ginibre}
K(x,y)=\rho e^{\pi\rho x \bar{y}} e^{-\frac{\pi\rho}{2}( |x|^2 + |y|^2)},
\quad 
 x,y \in W=\mathcal{B}(0,R). 
\end{equation}
This kernel is that of the usual Ginibre process defined, e.g., in~\cite{DecreusefondFlintVergne}, to which we have applied a homothety of parameter $\sqrt{\pi\rho}>0$ : $x\mapsto x/(\sqrt{\pi\rho})$. Here, $\rho>0$ is the density of the point process. Next we recall a few features of the Ginibre process. 
\begin{itemize} 
\item The intensity function of the Ginibre process 
 is 
 given by 
\begin{equation} 
\label{beq} 
\rho^{(1)}(x)=K(x,x)= \rho, 
\end{equation} 
 c.f.~\cite{ShiraiTakahashi}. This means that the average number of points in a bounded set $B\subset \mathcal{B}(0,R)$ is $\int_B \rho\,\mathrm{d}x\mathrm{d}y$.
\\ 
 
\item 
 The covariance of any 
 $\alpha$-DPPs of kernel $K$ is given by 
\begin{equation*}
\mathrm{Cov}(\mathcal{K}(A),\mathcal{K}(B))=\alpha\int_{A\times B}|K(x,y)|^2\,\mathrm{d}x\mathrm{d}y, 
\end{equation*}
 where 
 $\mathcal{K}(A)$ and $\mathcal{K}(B))$ 
 denote the random number of point process 
 points located within the disjoint bounded sets 
 $A,B\subset\mathbb{R}^2$. 
\\ 
 
\item 
 For every bounded set $B\subset\mathbb{R}^2$ we have 
\begin{equation}
\label{aeq} 
 Pr(\mathcal{K}\cap B = \emptyset ) = \mathrm{Det}(\mathrm{Id}+\alpha K_B)^{-1/\alpha},
\end{equation}
 where $K_B$ is the operator restriction of $K$ to the space $\mathrm{L}^2(B)$ of square integrable functions on $B$. Here, $\mathrm{Id}$ is the identity operator on $\mathrm{L}^2(B)$ and $\mathrm{Det}$ stands for the Fredholm determinant which is defined e.g. in~\cite{Brezis}. 
\end{itemize} 
Since $\alpha < 0$, $\mathcal{K}(A)$ and $\mathcal{K}(B)$ are negatively correlated and the associated $\alpha$-DPP is known to be locally Gibbsian, see e.g.~\cite{GeorgiiYoo}, therefore it is a kind of repulsive point process. In particular, the Ginibre $\alpha$-DPP exhibits more repulsion when $\alpha$ is close to $-1$. As $\alpha\rightarrow 0$, $\mathcal{K}(A)$ and $\mathcal{K}(B)$ tend to not be correlated, and the corresponding point process converges to the PPP, c.f.~\cite{ShiraiTakahashi}.

 We will write $\mathcal{K}\sim\mathrm{Det}(\alpha,K,\rho)$ 
 when $\mathcal{K}$ is an $\alpha$-DPP with 
 kernel $K$ defined in \eqref{eq:ginibre} and density with respect to the Lebesgue measure $\rho$. 
 The spectral theorem for Hermitian and compact operators 
 yields the following decomposition for the kernel of $K$:
\begin{equation*}
K(x,y)=\sum_{n\ge 0} \lambda_n \varphi_n(x)\overline{\varphi_n(y)},
\end{equation*}
where $(\varphi_i)_{i\ge 0}$ is a basis of $\mathrm{L}^2({O},\lambda)$, and $(\lambda_i)_{i\ge 0}$ the corresponding eigenvalues. 
 In e.g.~\cite{DecreusefondFlintVergne} it is shown that
 the eigenvalues of the Ginibre point process on $O=\mathcal{B}(0,R)$ 
 are given by 
\begin{equation}
\label{eq:eigenvalues}
\lambda_n = \frac{\Gamma(n+1, \pi\rho R^2)}{n!},
\end{equation}
 where 
\begin{equation}
\label{eq:defgamma}
\Gamma(z,a) := \int_0^a e^{-t} t^{z-1}\,\mathrm{d}t, 
\qquad 
z \in \mathbb{C}, \quad a \ge 0, 
\end{equation}
 is the lower incomplete Gamma function. 
 On the other hand, the eigenvectors of $K$ are given by
\begin{equation*}
\phi_n(z) :=  \frac{1}{\sqrt{\lambda_n}}\frac{\sqrt{\rho}}{\sqrt{ n!}} e^{-\frac{\pi\rho}{2} | z |^2} (\sqrt{\pi\rho} z)^n, 
\qquad 
 n \in \mathbb{N}, 
\quad 
 z \in O. 
\end{equation*}
 We refer to~\cite{DecreusefondFlintVergne} for 
 further mathematical details on the Ginibre point process. 
 
Lastly, we emphasize that the Ginibre $\alpha$-DPP is stationary, in the sense that its distribution is invariant with respect to translations, c.f.~\cite{DecreusefondFlintVergne}. Hence, our choice of $O=\mathcal{B}(0,R)$ centered at the origin instead of $\mathbf{x}_i$ is justified.

\subsection{Performance Metric} \label{sec:metrics}
We define the performance metrics of the sensor node as the expected energy harvesting rate, power outage probability and transmission outage probability. 
The expected RF energy harvesting rate is defined as:
\begin{equation}\label{eq:expect}
 \varphi \triangleq \mathbb{E} \left[ P^{i}_{\mathrm{H}} \right]	.
\end{equation}

An power outage occurs when the sensor node becomes inactive due to lack of enough energy supply. The power outage probability is then given by,
\begin{equation}
\label{eq:phi}
	\phi	=	Pr \left(	 P^i_{\mathrm{H}}	<	P^i_{\mathrm{C}}	\right)
\end{equation}

Let $m\ge 0$ denote the minimum transmission rate requirement. If the sensor fails to achieve this requirement, transmission outage occurs. The transmission outage probability can be defined as,
\begin{equation}
\label{eq:psi}
	\psi	=	Pr \left(	C_i	<	m	\right)
\end{equation}


\section{Performance Analysis}
\label{sec:Analysis}
\subsection{Ginibre $\alpha$-determinantal point process}
In this section we estimate the metrics defined in Section~\ref{sec:metrics}
when $\mathcal{K}\sim\mathrm{Det}(\alpha,K,\rho)$ is the Ginibre DPP with parameter $\alpha=-\frac{1}{m}$, where $m\in\mathbb{N}^*$. 
We assume additionally that $P^k_{\mathrm{S}}:=P_{\mathrm{S}}$, $G^{k}_{\mathrm{S}}:=G_{\mathrm{S}}$, and $\lambda_k:=\lambda$ do not depend on $k$.

For the estimation of \eqref{eq:phi} and \eqref{eq:psi}, we might proceed by Monte Carlo simulation of the underlying $\alpha$-DPP. Simulation of $\alpha$-DPPs when $\alpha <0$ is done by using the Schmidt orthogonalization algorithm developed in full generality in~\cite{Hough}, and specifically in~\cite{DecreusefondFlintVergne} for the Ginibre point process. The simple generalization to all $\alpha <0$ can be found in the recent survey~\cite{DecreusefondFlintPrivaultTorrisi}, and additional details on DPP can be found in~\cite{DecreusefondFlintPrivaultTorrisi2}. 

Monte Carlo methods can however be quite time-consuming in practice, 
especially when it is repeatedly applied to multiple values of the parameters.
Thus, in many applications, it is of major interest to have some practical bounds at hand, such as the ones which we present now. First, we obtain the expected RF energy harvesting rate in the following theorem.
\begin{theorem}
The expected RF energy harvesting rate can be explicitly computed as 
\begin{multline}
\label{eq:evenergy}
\mathbb{E}[P_H^i] =2 \pi \beta_{i}  P_{\mathrm{S}} \frac{G^{k}_{\mathrm{S}}G^{i}_{\mathrm{H}} \lambda^{2}}{(4\pi )^{2}} \\
\rho\left(\frac{\epsilon}{R+\epsilon}+\ln(R+\epsilon)-1-\ln(\epsilon)\right),
\end{multline}
\end{theorem}
\begin{proof}
We have
\begin{equation*}
\mathbb{E}[P_H^i]=\beta_{i}  P_{\mathrm{S}} \frac{G^{k}_{\mathrm{S}}G^{i}_{\mathrm{H}} \lambda^{2}}{(4\pi )^{2}} \int_{W} \frac{\rho^{(1)}(x)}{(\epsilon+\|x\|)^2}\,\mathrm{d}x
\end{equation*}
by Campbell's formula~\cite{Kallenberg}, where 
$\rho^{(1)}(x)=K(x,x)=\rho$ 
 is the intensity function of $\mathcal{K}$ given by \eqref{beq},
 which yields 
\begin{equation*}
\mathbb{E}[P_H^i]=\beta_{i}  P_{\mathrm{S}} \frac{G^{k}_{\mathrm{S}}G^{i}_{\mathrm{H}}  \lambda^{2}}{(4\pi )^{2}} 2 \pi \int_{0}^R \rho \frac{r}{(\epsilon+r)^2}\,\mathrm{d}r,
\end{equation*}
 by polar change of variable,
 and the integral on the r.h.s. is computed explicitly as
\begin{equation*}
 \int_{0}^R\frac{r}{(\epsilon+r)^2}\,\mathrm{d}r=
\left(\frac{\epsilon}{R+\epsilon}+\ln(R+\epsilon)-1-\ln(\epsilon)\right),
\end{equation*}
which yields the result.
\end{proof}

Note here that the expected RF energy harvesting rate does not depend on the parameter $\alpha$ of the DPP.

Now, we give a practical upper bound to the probability that the sensor node becomes inactive due to lack of energy supply $Pr \left(	 P^i_{\mathrm{H}}	<	P^i_{\mathrm{C}}	\right)$. Specifically, we prove the following:
\begin{theorem}
\label{thm:estimationphi}
 Let us define 
\begin{eqnarray}
\gamma_i:=\sqrt{ \frac{\beta_{i} P_{\mathrm{S}}G^{k}_{\mathrm{S}} G^{i}_{\mathrm{H}}  \lambda^{2}}{(4\pi)^{2}P^i_{\mathrm{C}} } }.
\end{eqnarray}
Then, the following bound holds:
\begin{align}
& Pr \left(	 P^i_{\mathrm{H}}	<	P^i_{\mathrm{C}}	\right) \nonumber \\
& \le  \left(\prod_{n\ge 0} (1+\alpha  \frac{\Gamma(n+1, \pi\rho\inf(R,\gamma_i)^2)}{n!})\right)^{-1/\alpha}, \label{dpp_po}
\end{align}
where $\Gamma(z,a)$ is defined in \eqref{eq:defgamma}.
\end{theorem}
\begin{proof}
Let us define
\begin{equation*}
f(\mathbf{x}_k):=\beta_{i} P_{\mathrm{S}} \frac{G^{k}_{\mathrm{S}} G^{i}_{\mathrm{H}} \lambda^{2}}{(4\pi (\epsilon+\|\mathbf{x}_k\|))^{2}},
\end{equation*}
for $k\in\mathcal{K}$. 
\begin{align*}
Pr \left(	 P^i_{\mathrm{H}}	<	P^i_{\mathrm{C}}	\right) &=Pr(\sum_{k\in\mathcal{K}} f(k) \le P^i_{\mathrm{C}})\\
	&\le Pr(\forall k \in\mathcal{K},\ f(\mathbf{x}_k) \le P^i_{\mathrm{C}})\\
	&= Pr(\forall k \in\mathcal{K},\ \|\mathbf{x}_k\| \ge \gamma_i -\epsilon )\\
	&= Pr(\mathcal{K}\cap\mathcal{B}(0,\gamma_i-\epsilon) = \emptyset ),
\end{align*}
where we have chosen $\epsilon$ such that $\gamma_i-\epsilon\ge 0$. 
Thus by \eqref{aeq} we obtain
\begin{equation*}
Pr \left(	 P^i_{\mathrm{H}}	<	P^i_{\mathrm{C}}	\right) \le \mathrm{Det}(\mathrm{Id}+\alpha K_{\mathcal{B}(0,\gamma_i-\epsilon)})^{-1/\alpha}.
\end{equation*}
Since in our case $K$ is the Ginibre kernel, the eigenvalues of $K$ are given by \eqref{eq:eigenvalues}. By standard properties of the Fredholm determinant which can be found e.g. in~\cite{Brezis}, we find
\begin{multline*}
Pr \left(	 P^i_{\mathrm{H}}	<	P^i_{\mathrm{C}}	\right) \\
\le  \left(\prod_{n\ge 0} \left(1+\alpha  \frac{\Gamma(n+1, \pi\rho \inf(R,\gamma_i-\epsilon)^2)}{n!}\right)\right)^{-1/\alpha},
\end{multline*}
and the result follows by letting $\epsilon$ go to zero on the r.h.s. of the previous equation, since the associated function of $\epsilon$ is continuous.
\end{proof}

It should be noted that the eigenvalues in Theorem~\ref{thm:estimationphi} are in decreasing order, and decrease exponentially when $n \ge  \pi\rho \inf(R,\gamma_i)^2$, see~\cite{DecreusefondFlintVergne} for details. Hence, the product which appears in Theorem~\ref{thm:estimationphi} is very well approximated by 
\begin{eqnarray}
  \left(\prod_{n\ge 0}^N \left(1+\alpha  \frac{\Gamma(n+1,  \pi\rho\inf(R,\gamma_i)^2)}{n!}\right)\right)^{-1/\alpha}, \nonumber
\end{eqnarray}
 where $N \gg \pi\rho\inf(R,\gamma_i)^2$.
 
The variations of the bound obtained in Theorem~\ref{thm:estimationphi} with respect to $\alpha$ is now explicited. It is easy to show that
\begin{eqnarray}
 \frac{\mathrm{d}}{\mathrm{d}\alpha} \ln\left[\left(\prod_{n\ge 0} (1+\alpha \lambda_n)\right)^{-1/\alpha}\right] \nonumber \\
	=  \frac{1 }{\alpha^2}  \sum_{n\ge 0} \frac{(1+\alpha \lambda_n)\ln(1+\alpha \lambda_n) - \alpha \lambda_n}{1+\alpha \lambda_n} \ge 0, \nonumber
 \end{eqnarray}
 which means that the bound of Theorem~\ref{thm:estimationphi} is lowest when $\alpha=-1$, i.e. when repulsion is maximal, and increases with $\alpha$. 

Next, in order to estimate the transmission outage probability $Pr \left(	R_i	<	m	\right)$, it suffices to apply Theorem~\ref{thm:estimationphi}. Specifically, we have
\begin{theorem}
Let we define
\begin{equation}
\label{eq:gammam}
	\gamma_i^m:=\sqrt{ \frac{\beta_{i} P_{\mathrm{S}}G^{k}_{\mathrm{S}} G^{i}_{\mathrm{H}} \lambda^{2}}{(4\pi)^{2}\left(P^i_{\mathrm{C}}+ \frac{\sigma^2}{h_i}\left(2^m-1\right)\right) } },
\end{equation}
then we obtain
\begin{align}
\label{eq:boundpsi}
	Pr \left(	R_i	<	m	\right) &= Pr\left( \left[	P^i_{\mathrm{H}} - P^i_{\mathrm{C}} \right]^+ < \frac{\sigma^2}{h_i}\left(2^m-1\right) \right)\nonumber\\
	&= Pr\left( P^i_{\mathrm{H}}  <P^i_{\mathrm{C}}+ \frac{\sigma^2}{h_i}\left(2^m-1\right) \right)\nonumber\\
	&\le  \left(\prod_{n\ge 0} \left(1+\alpha  \frac{\Gamma(n+1,  \pi\rho\inf(R,\gamma_i^m)^2)}{n!}\right)\right)^{-1/\alpha},
\end{align}
where we have used Theorem~\ref{thm:estimationphi} and the fact that $2^m-1\ge 0$.
\end{theorem}

\begin{figure*}[t]
\begin{center}
$\begin{array}{ccc} 
\epsfxsize=2.3 in \epsffile{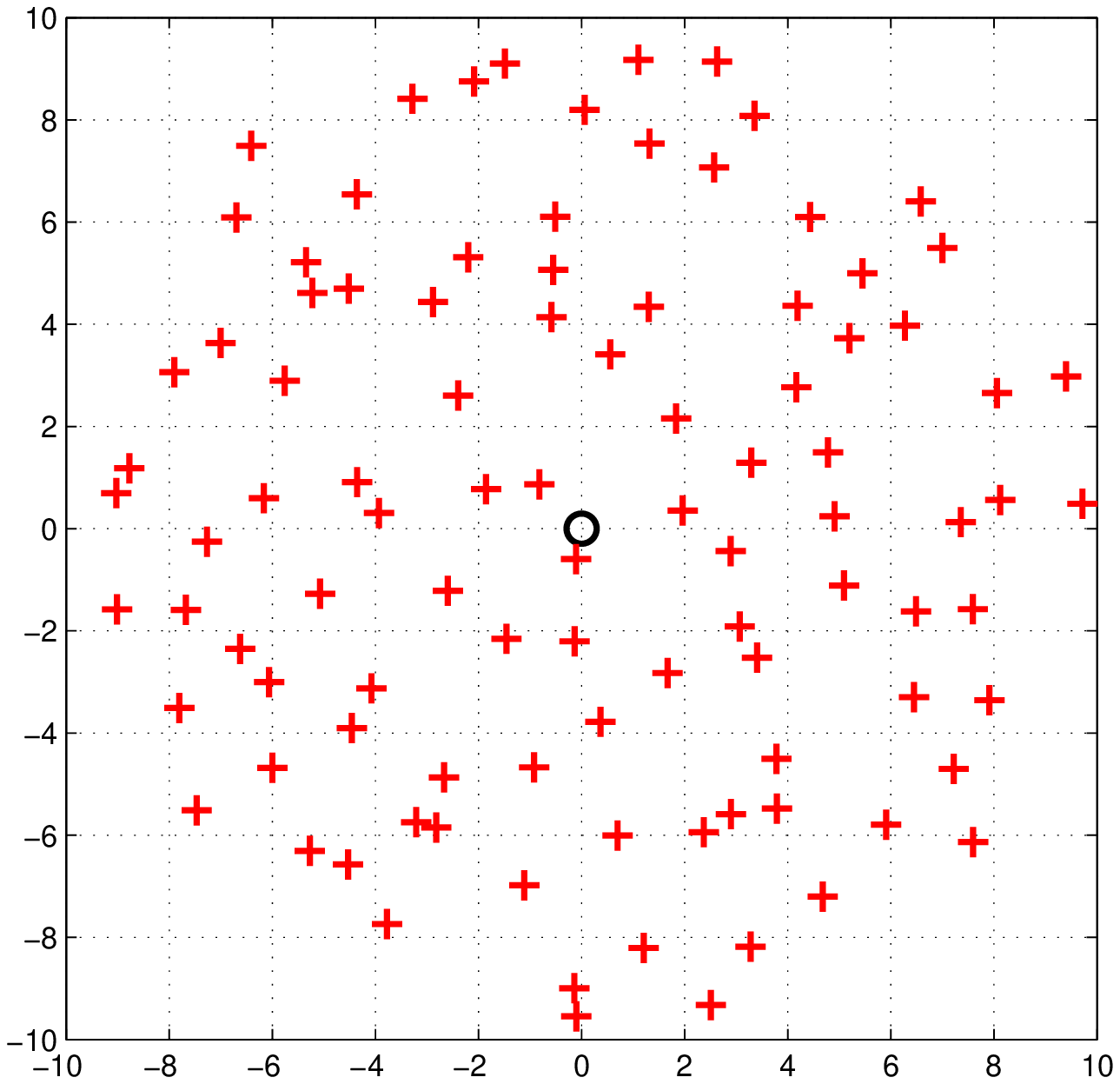}	&
\epsfxsize=2.3 in \epsffile{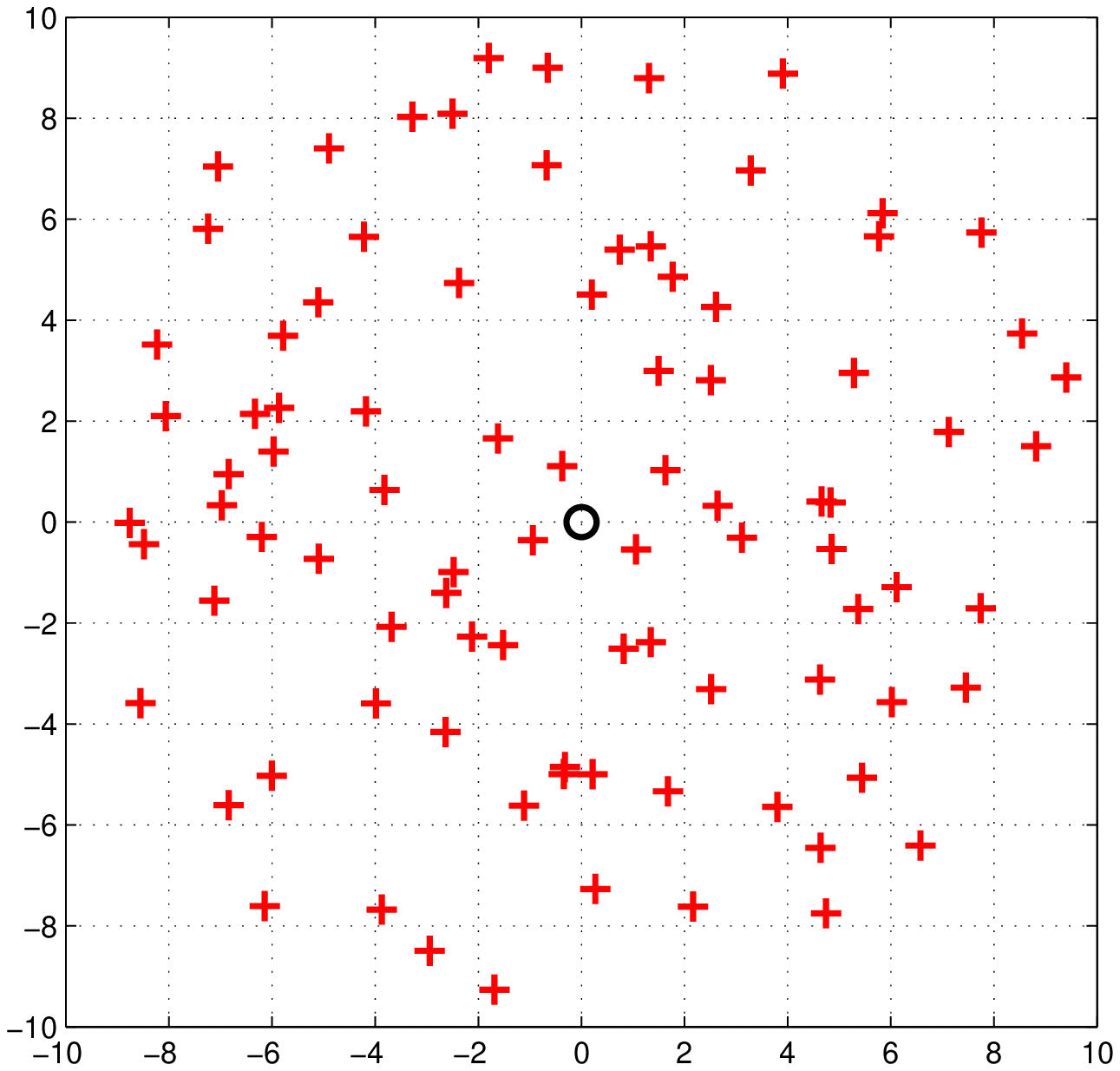}	&
\epsfxsize=2.3 in \epsffile{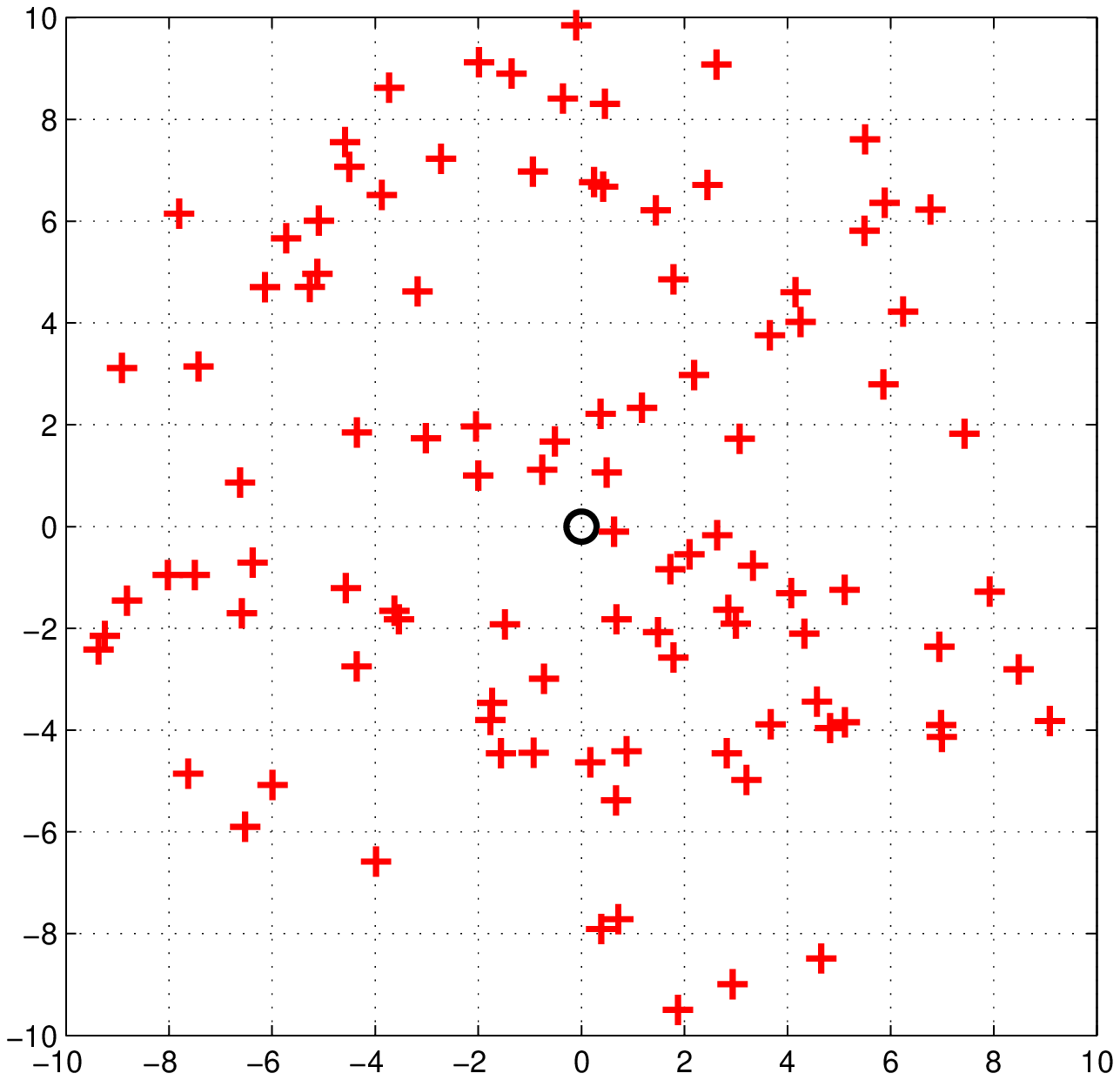}	\\ [-0.2cm]
(a)	& (b)	&	(c)
\end{array}$
\caption{Snapshots of the distribution of ambient RF energy sources (a) $\alpha=-1$ (b) $\alpha=-0.5$ (c) $\alpha=-0.03$.}
\label{Distribution}
\end{center}
\end{figure*} 

\subsection{Poisson point process}
We end this section by computing the previous results in the case of PPP. 
As a corollary of Theorem~\ref{thm:estimationphi}, we find in the case of a PPP (which is obtained as the limit as $\alpha\rightarrow 0$ in the theorem):
\begin{corollary}
Let $\mathcal{K}\sim\mathrm{Poiss}(O,\rho)$ be a Poisson process on $W=\mathcal{B}(0,R)$ with density $\rho$. Then, the following bound holds:
\begin{eqnarray} \label{ppp_po} 
Pr \left(	 P^i_{\mathrm{H}}	<	P^i_{\mathrm{C}}	\right) \le e^{-\pi\rho\inf(R,\gamma_i)^2},
\end{eqnarray}
where $\gamma_i$ is as defined in Theorem~\ref{thm:estimationphi}.
\end{corollary}

Similarly, we have, 
\begin{corollary}
When $\mathcal{K}$ is a Poisson process, the transmission outage probability can be estimated as follows:
\begin{eqnarray} \label{ppp_to}
	Pr \left(	R_i	<	m	\right)\le e^{-\pi\rho\inf(R,\gamma_i^m)^2},
\end{eqnarray}
where $\gamma_i^m$ is defined in \eqref{eq:gammam}. 
\end{corollary}

\section{Numerical Results}
The results in this section are obtained based on the following values of parameters unless specified otherwise. Both the transmitting antenna gain $G^{k}_{S}$ and receiving antenna gain $G^{i}_{S}$ are set to be $1.5$. The RF-to-DC power conversion efficiency $\beta_{i}$ is considered to be $30\%$. We assume that all the ambient RF energy sources are LTE-enabled mobiles operating with $P^{k}_{S}=1W$ transmit power on the typical $1800MHz$ frequency. The corresponding wavelength $\lambda_{k}$ is $0.167m$. The circuit power consumption $P^{i}_{C}$ is fixed to be $-18dBm$ (i.e., $15.8\mu W$) as in~\cite{NParks}. The noise power is $-90dBm$, (i.e., $10^{-12}W$). The transmission bandwidth is set to be $1kHz$. The channel gain between the sensor and data sink is calculated as \cite{XTVTlu}: $h_{i}=62.5d^{-4}$, where $d$, assumed to be $50$ meters, is the distance between the sensor node and the data sink.


\begin{figure*}[t]
\begin{center}
$\begin{array}{ccc} 
\epsfxsize=3.5 in \epsffile{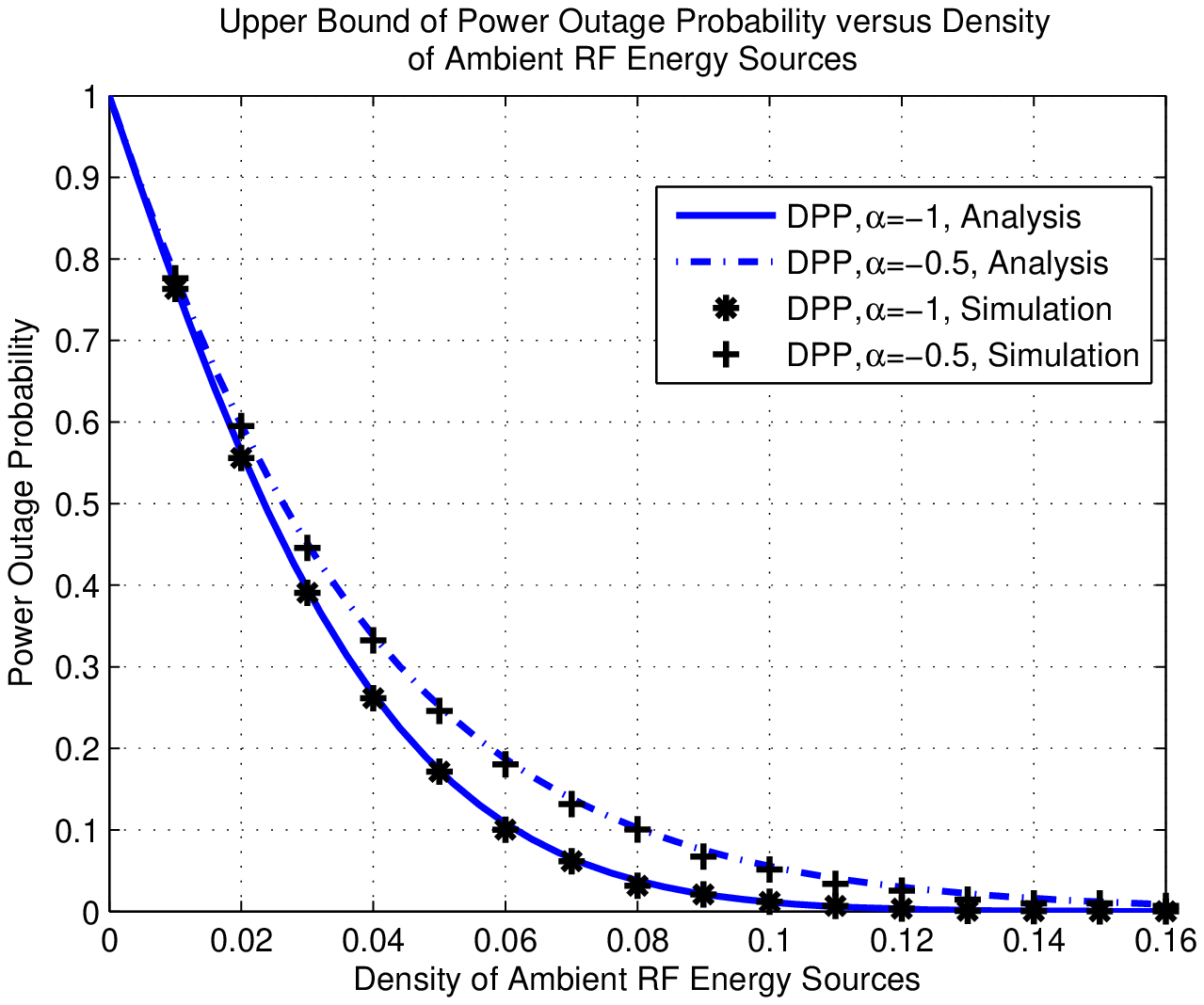}	&
\epsfxsize=3.5 in \epsffile{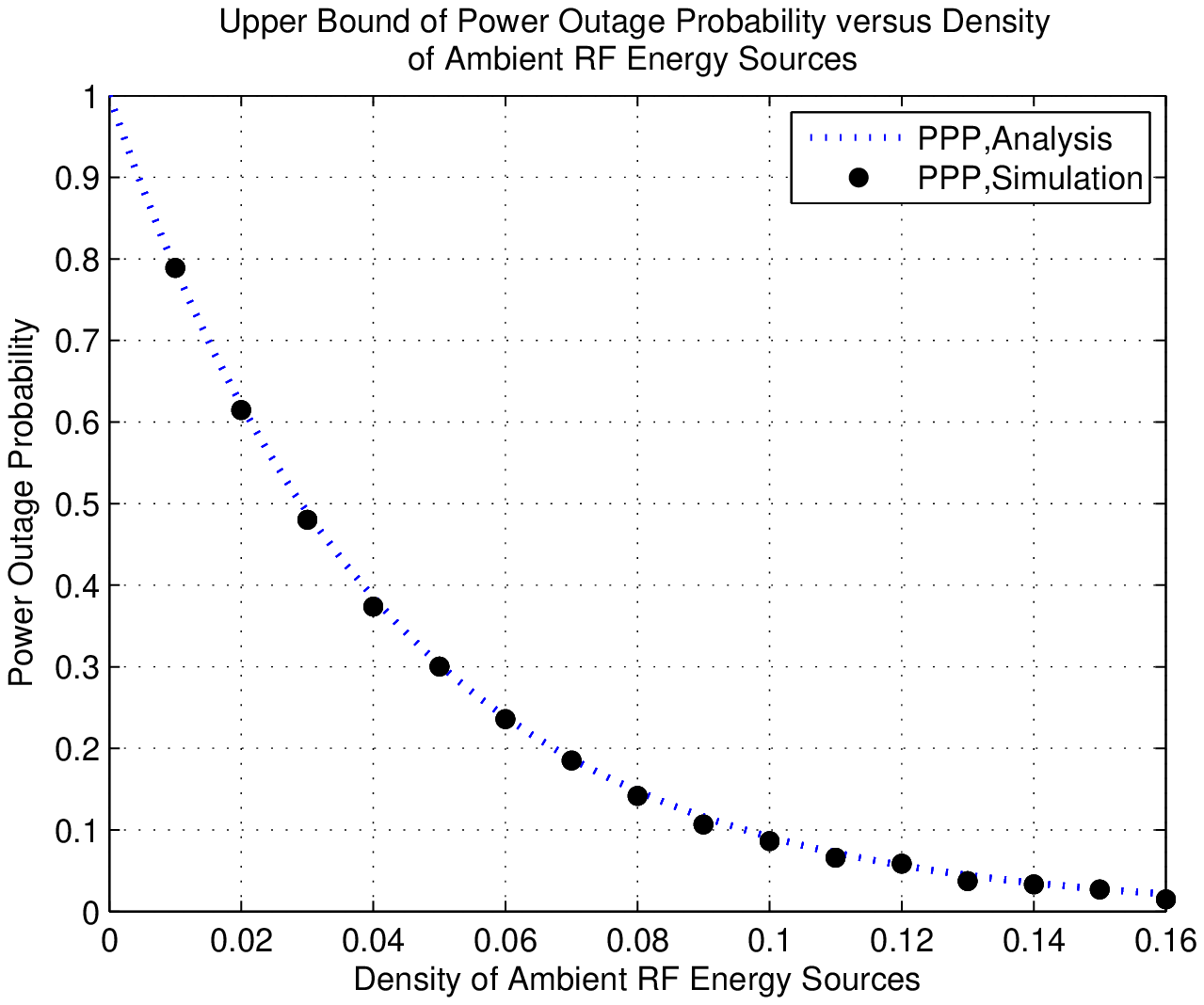}	\\ [-0.2cm]
(a)	& (b) 
\end{array}$
\caption{Upper bound of an power outage probability versus density of ambient RF energy sources for (a) DPP and (b) PPP.}
\label{Power_outage}
\end{center}
\end{figure*} 

\begin{figure*}[t]
\begin{center}
$\begin{array}{ccc} 
\epsfxsize=3.4 in \epsffile{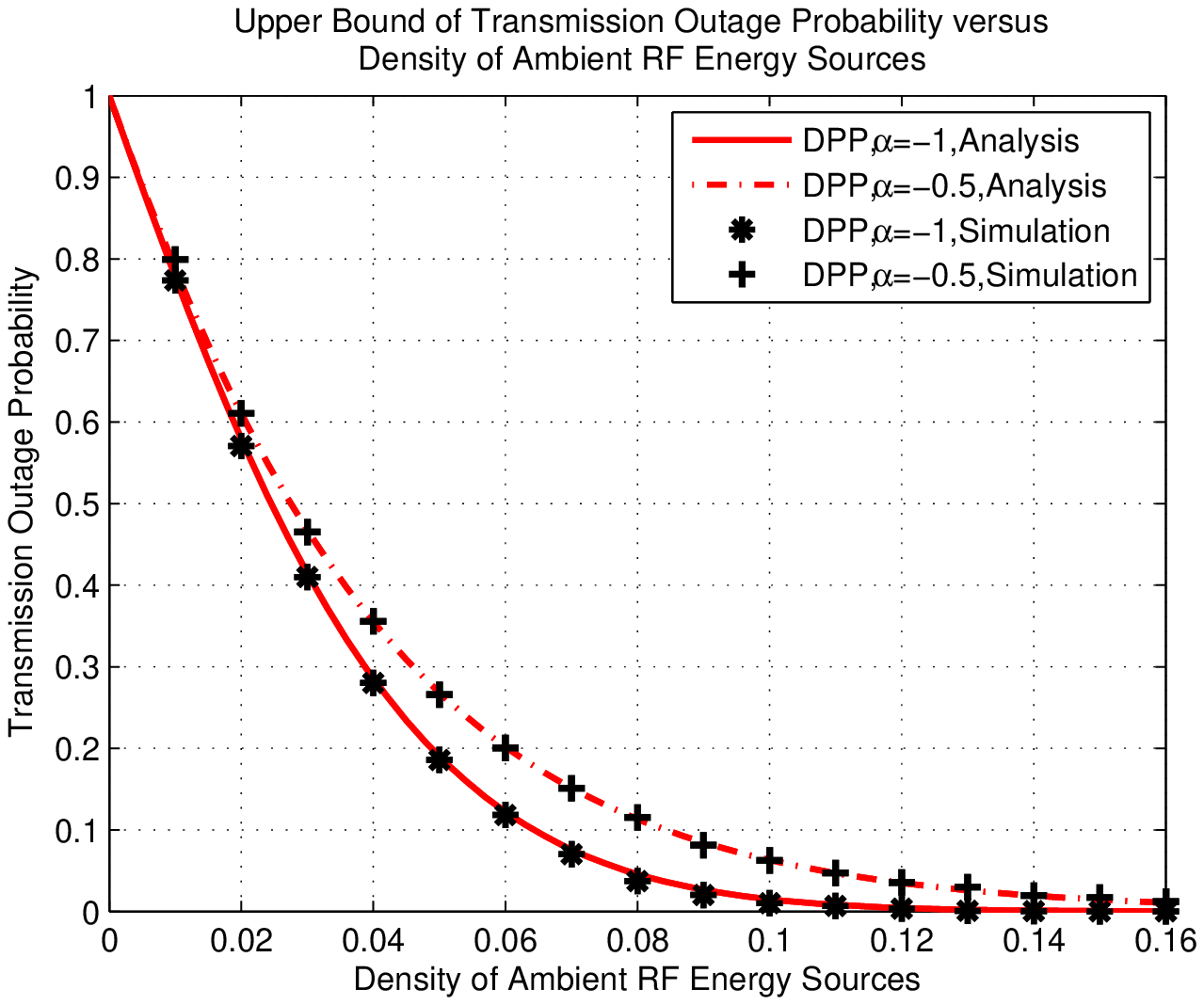}	&
\epsfxsize=3.4 in \epsffile{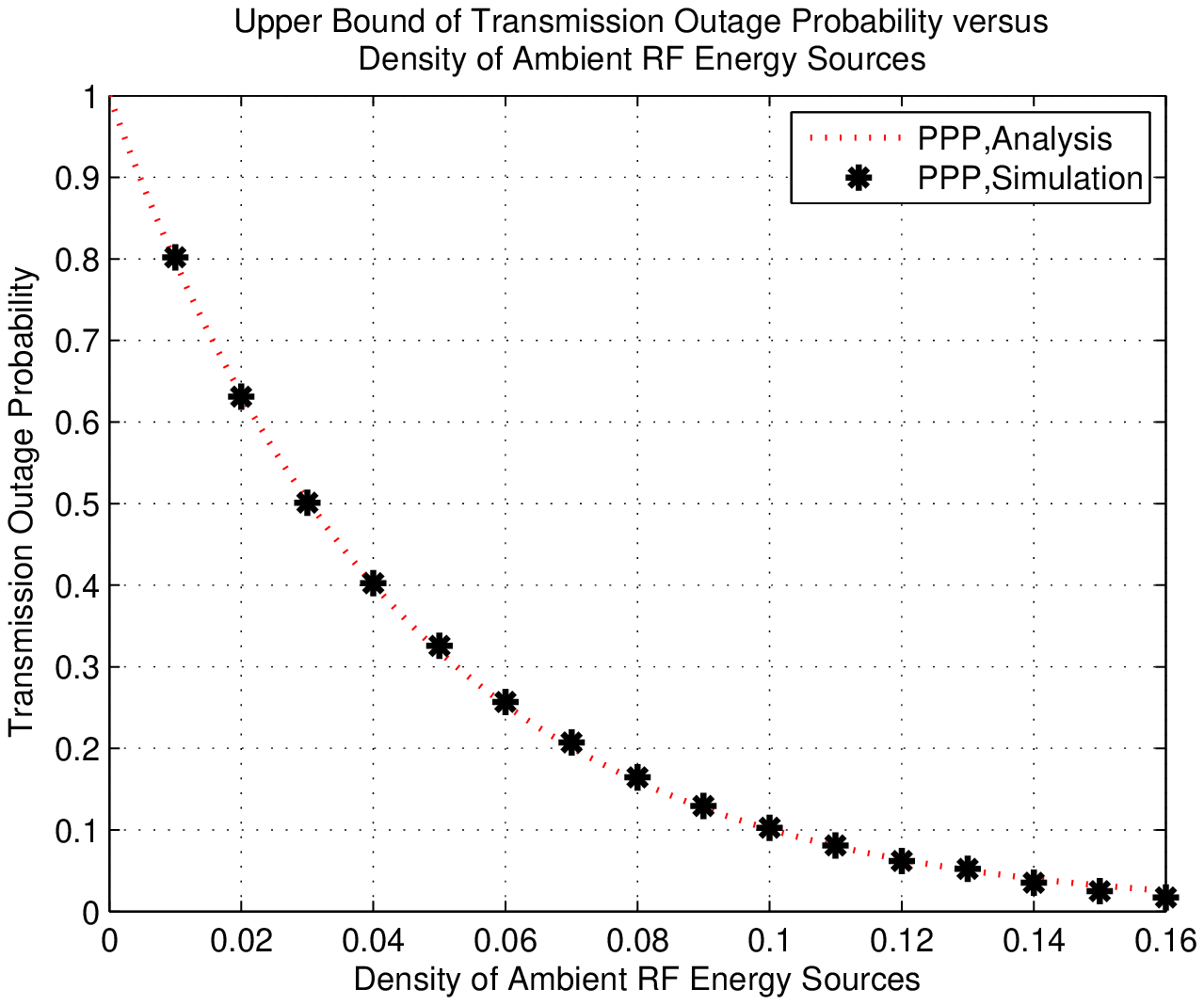}	\\ [-0.2cm]
(a)  	& (b)   
\end{array}$
\caption{Upper bound of transmission an outage probability versus density of ambient RF energy sources for (a) DPP  and (b) PPP ($m=3$).}
\label{Transmission_outage}
\end{center}
\end{figure*} 

In Fig.~\ref{Distribution}, we show some snapshots of the distribution of ambient RF sources in a circle area with the radius $R=10m$, when the RF source density is $1$. It is illustrated that strong repulsion exists between the RF sources when $\alpha=-1$. As a result, the RF sources tend to locate uniformly over the area. We can observe that the repulsion keeps decreasing when $\alpha$ is more approaching zero. When $\alpha=-0.03$, the RF sources 
show clustering, which is a feature of PPP.

Figure~\ref{Power_outage} shows the upper bound of power outage probability versus the density of ambient RF sources. It is verified that the analytical expressions in (\ref{dpp_po}) and (\ref{ppp_po}) are very tight. It can be observed in Fig.~\ref{Power_outage}(a) that the upper bound of an power outage probability increases with $\alpha$. Compared between Fig.~\ref{Power_outage}(a) and Fig.~\ref{Power_outage}(b), we observe that when the DPP tends to be a PPP, i.e. repulsion becomes reduced to a minimum, the probability that the sensor becomes inactive is higher. The power outage probability is the lowest when $\alpha=-1$, i.e. when the model exhibits maximal repulsion. 

We demonstrate the upper bound of transmission outage probability versus the density of ambient RF sources in Fig.~\ref{Transmission_outage}. It can be seen that the analytical expressions in (\ref{eq:boundpsi}) and (\ref{ppp_to}) are very accurate. Compared between Fig.~\ref{Transmission_outage}(a) and~\ref{Transmission_outage}(b), we can make a similar observation that when the DPP tends to be a PPP, the probability that the sensor fails to reach the minimum rate requirement is higher. Finally, we can conclude that given a certain density of ambient RF sources, the sensor achieves better performance when the distribution of ambient RF sources shows stronger repulsion/less attraction.


\section{Conclusion}

This paper has presented the performance analysis of a wireless sensor powered by ambient RF energy through adopting a stochastic-geometry approach. We have analyzed the cases when the ambient RF sources are geographically distributed as a Ginibre $\alpha$-DPP, whose special case is a Poisson point process. We have derived the expression of expected RF energy harvesting rate. We have also characterized the worst-case performance of a sensor node in terms of the upper bounds of power outage and transmission outage probabilities. Numerical results have shown that all the simulation results agree with the corresponding analytical results, which leads us to believe that the upper-bounds are usable in practice. We have additionally found that given a certain density of ambient RF sources, the sensor achieves better performance when the distribution of ambient RF sources shows stronger repulsion/less attraction. Our future work will extend the performance analysis from an individual node level to a system level.

\end{document}